\documentclass[12pt]{amsart}
\usepackage{amssymb,amsthm,amsmath,enumerate}
\usepackage[numbers,sort&compress]{natbib}
\usepackage{color,hyperref,url}
\usepackage{graphicx}
\usepackage{amssymb,amsthm,amsmath}
\usepackage[numbers,sort&compress]{natbib}
\usepackage{color}
\usepackage{graphicx}
\usepackage{tikz}

\title[]{Sharp spectral transition for embedded eigenvalues of perturbed periodic Dirac operators}

\author{Kang Lyu}
\address{ School of Mathematics and Statistics, Nanjing University of Science and Technology, Nanjing, 210094, Jiangsu, People’s Republic of China}
\email{lvkang201905@outlook.com}
\urladdr{https://kanglyu.github.io/index.html}

\author{Chuanfu Yang}
\address{ School of Mathematics and Statistics, Nanjing University of Science and Technology, Nanjing, 210094, Jiangsu, People’s Republic of China}
\email{chuanfuyang@njust.edu.cn}



\newcommand{\abs}[1]{\left\lvert #1 \right\rvert}

\theoremstyle{plain}
\newtheorem{theorem}{Theorem}[section]
\newtheorem{corollary}[theorem]{Corollary}
\newtheorem{lemma}[theorem]{Lemma}
\newtheorem{proposition}[theorem]{Proposition}

\newcommand{\Z}{\mathbb{Z}}
\theoremstyle{definition}

\begin{document}
	{\begin{abstract} We consider the Dirac equation on $L^2(\mathbb{R})\oplus L^2(\mathbb{R})$
			\begin{align}
				Ly=
				\begin{pmatrix}
					0&-1\\
					1&0
				\end{pmatrix}
				\begin{pmatrix}
					y_1\\
					y_2
				\end{pmatrix}'+
				\begin{pmatrix}
					p&q\\
					q&-p
				\end{pmatrix}\begin{pmatrix}
					y_1\\
					y_2
				\end{pmatrix}+
				V\begin{pmatrix}
					y_1\\
					y_2
				\end{pmatrix}=\lambda y,\nonumber
			\end{align} 
			where $y=y(x,\lambda)=\tbinom{y_1(x,\lambda)}{y_2(x,\lambda)}$,  $p$ and $q$ are real $1$-periodic, and 
			\begin{align}
				V=\begin{pmatrix}
					V(x)&0\\
					0&-V(x)
				\end{pmatrix}\nonumber
			\end{align}
			is the perturbation which satisfies $V(x)=o(1)$ as $\abs{x}\to\infty.$
			Under such perturbation, the essential spectrum of $L$ coincides with that there is no perturbation.
		We prove that if $V(x)=\frac{o(1)}{1+\abs{x}}$ as $x\to\infty$ or $x\to-\infty$, then there is no embedded eigenvalues (eigenvalues appear in the essential spectrum). For any given finite set inside of  the essential spectrum which satisfies the non-resonance assumption, we construct smooth potentials with $V(x)=\frac{O(1)}{1+\abs{x}}$ as $\abs{x}\to\infty$ so that the set becomes embedded eigenvalues. For any given countable set inside of the essential spectrum which satisfies the non-resonance assumption, we construct smooth potentials with $V(x)<\frac{\abs{h(x)}}{1+\abs{x}}$ as $\abs{x}\to\infty$ so that the set becomes embedded eigenvalues, where $h(x)$ is any given function with $\lim_{x\to\pm\infty}\abs{h(x)}=\infty.$
	\end{abstract}}
	
	\maketitle
	\section{Introduction and main results}
	Consider the Dirac equation on $L^2(\mathbb{R})\oplus L^2(\mathbb{R})$
\begin{align}\label{perturbeddiracequation}
	Ly=
\begin{pmatrix}
	0&-1\\
	1&0
\end{pmatrix}
\begin{pmatrix}
	y_1\\
	y_2
\end{pmatrix}'+
\begin{pmatrix}
	p&q\\
	q&-p
\end{pmatrix}\begin{pmatrix}
y_1\\
y_2
\end{pmatrix}+
V\begin{pmatrix}
	y_1\\
	y_2
\end{pmatrix}=\lambda y,
\end{align} 
where $y=y(x,\lambda)=\tbinom{y_1(x,\lambda)}{y_2(x,\lambda)}$,  $p$ and $q$ are real $1$-periodic, and 
\begin{align}
	V=\begin{pmatrix}
		V(x)&0\\
		0&-V(x)
	\end{pmatrix}\nonumber
\end{align}
is the perturbation. In the following, we always assume that 
\begin{align}\label{assumptiondiracpotential}
	V(\cdot)\in L^\infty,\ V(x)=o(1), \abs{x}\to\infty.
\end{align}
When $V\equiv 0$, we have an unperturbed $1$-periodic Dirac equation
\begin{align}\label{periodicdiracequation}
	L_0g=
	\begin{pmatrix}
		0&-1\\
		1&0
	\end{pmatrix}
	\begin{pmatrix}
		g_1\\
		g_2
	\end{pmatrix}'+
	\begin{pmatrix}
		p&q\\
		q&-p
	\end{pmatrix}\begin{pmatrix}
	g_1\\
		g_2
	\end{pmatrix}
	=\lambda g,
\end{align} 
where $g=g(x,\lambda)=\tbinom{g_1(x,\lambda)}{g_2(x,\lambda)}$. 

The spectrum of $L_0$ consists of a class of intervals:
\begin{align}
	\sigma(L_0)=\sigma_{ess}(L_0)=\cup_{j=-\infty}^\infty [a_j,b_j],\nonumber
\end{align}
where $a_j,b_j$ are eigenvalues of $L_0g=\lambda g,\ x\in (0,1)$ with periodic or anti-periodic boundary conditions. Under the assumption \eqref{assumptiondiracpotential}, by Weyl's criterion, one has
\begin{align}
	\sigma_{ess}(L)=\sigma_{ess}(L_0)=\cup_{j=-\infty}^\infty [a_j,b_j].\nonumber
\end{align}
For any $\lambda\in\sigma_{ess}(L_0)$, let $g(x,\lambda)$ be the corresponding Floquet solution with 
\begin{align}\label{Floquetexponet}
	g(x+1,\lambda)=e^{ik(\lambda)}g(x,\lambda),
\end{align}
where $k(\lambda)$ is the quasimomentum. By Floquet theory, $k(\lambda)$ monotonically decreases from $\pi$ to $0$ or monotonically increases from $0$ to $\pi$ on each interval $[a_j,b_j]$ ( for example, see \cite{brown2013periodicbook}).

The problem of embedded eigenvalues into the essential spectrum or absolutely continuous spectrum has been widely studied on Schr\"odinger operators. For example, for the Schr\"odinger operator on $L^2(0,\infty)$, given by
\begin{align}
	Hu=-u''+Vu.\nonumber
\end{align}
Kato's result \cite{kato1959} shows that there is no embedded eigenvalues larger than $a^2$, where $a=\limsup_{x\to\infty}\abs{xV(x)}$.  Classical Wigner-von Neumann type potential 
\begin{align}
	V(x)=\frac{a}{1+x}\sin \left(kx+\theta\right)\nonumber
\end{align}
provides with a decaying oscillatory perturbation which creates a single embedded eigenvalue \cite{von1929}. Wigner-von Neumann type potentials play an important role in problems of spectral analysis (\cite{Lukicjspwignervonneumann,Lukictransperiodic,Lukiccmpdecayingoscillatory,judgejacobineumann,janasdiscreteschrodinger,simonovzerosspectraldensity,simonovspectralanalysis}), even the sharp bound for single embedded eigenvalue is obtained by the sign function of Wigner-von Neumann type potential (\cite{atkinson1978bounds,liudiscrete}).

For many embedded eigenvalues, Naboko \cite{Na86} constructed  potentials with countably many embedded (rationally independent) eigenvalues, and Simon \cite{simondense} improved the result with no restrictions on eigenvalues. In recent years, this direction has more fertile results \cite{kiselevperiodicmethod,liuirreducibilityfermivariety,krugerembeddedeigenvalues,liustarkmathN,liuasymptotic,kiselevsingularJAMS,vishwamdirac,judgegeometricapproach,judgediscretelevinsontechnique,JL19,Liuresonant,LO17,L19stark,liudiscrete}.
For example, in \cite{judgegeometricapproach,judgediscretelevinsontechnique}, Judge, Naboko, and Wood introduced perturbations that create embedded eigenvalues for Jacobi operators. In \cite{JL19},
Jitomirskaya and Liu introduced a novel idea to construct embedded eigenvalues for Laplacian on manifolds and this approach has been developed for various operators in \cite{Liuresonant,LO17,L19stark}. 

Consider the corresponding periodic problem of Schr\"odinger operator on $L^2(\mathbb{R})$, given by
\begin{align}
	Hu=H_0u+V_0u=-u''+q_0u+V_0u=\lambda u,
\end{align}
where $q_0$ is a real $1$-periodic function and $V_0$ is the perturbation. If there is no perturbation. Namely, $V_0=0$, then the spectrum of $H=H_0$ is a union of intervals
\begin{align}
	\sigma_{ess}(H_0)=\sigma(H_0)=\cup_{j=1}^\infty[c_j,d_j],\nonumber
\end{align}
where $c_j,d_j$ are eigenvalues of $H_0u=\lambda u,\ x\in (0,1)$ with periodic or anti-periodic boundary conditions.  The authors in \cite{LO17} constructed $V_0\in C^\infty$ to embed finitely (countably) many (non-resonant) eigenvalues into the essential spectrum. We provide an overview of their construction process. In contrast to the differential equations of the (modified) Pr\"ufer variables in \cite{atkinson1978bounds,L19stark,liudiscrete}, the authors need to handle a new equation coupled with a periodic term (one may obtain the speed of the decrease of eigensolutions by analyzing this differential equation, which becomes very complicated when there is a periodic term). Dealing with the Fourier expansion of the periodic term, the authors obtained a robust estimate for eigensolutions. After that, they can construct potentials so that the eigensolutions of different eigenvalues decrease one by one (when one of the eigensolutions decreases, others do not undergo significant increments), these are ensured by \eqref{RpmxC} and \eqref{Rj1.5R0}.
This exploration motivates our interest in investigating periodic Dirac operators. Once we can construct perturbations with \eqref{RpmxC} and \eqref{Rj1.5R0}, following the construction in \cite{LO17,JL19,L19stark}, we can handle the problem.  See \cite{korotyaevdislocationperiodicdirac,korotyaevperiodicdiracdubrovinequation,teschlperiodicbandedge} for more discussions of periodic operators.

The difficulty is that there is no proper Pr\"ufer transformation for periodic Dirac operators at hand. We mention that the classical Pr\"ufer transformation appearing in \cite{vishwamdirac} can only deal with the case that there is no periodic term. 
Different from the differential equations for (modified) Pr\"ufer variables ($R,\theta$) for periodic Schr\"odinger operators in \cite{kiselevperiodicmethod}, the differential equations for (modified) Pr\"ufer variables ($R,\theta_1,\theta_2$) for periodic Dirac operators are much more complicated. First of all, there is only one Pr\"ufer angle of periodic Schr\"odinger operators, but 
there are two Pr\"ufer angles for periodic Dirac operators and
the two Pr\"ufer angles induce one more differential equation ((c) of Lemma \ref{lemmaprufer} ) compared with the Schr\"odinger case. This causes that the differential equation for $R$ couples with two Pr\"ufer angles ((b) of Lemma \ref{lemmaprufer} ), so we need to consider one more parameter when estimating the eigensolution. Moreover, the differential equation of $R$ in \cite{LO17} is a direct product of a periodic function and $\frac{\sin\theta(x)}{x}$ ($\approx \frac{\sin kx}{x}$), so that once they obtain the estimate of the product of a periodic function and $\frac{\sin kx}{x}$, they can show Lemmas \ref{lemma-100} and \ref{lemma1.5}. However, for periodic Dirac operators, we need to estimate the integral of the product of a periodic function and $\left(\frac{g_1(x)\sin\theta_1(x)}{x}+\frac{g_2(x)\sin\theta_2(x)}{x}\right)$ at the same time. This requires a more delicate analysis.
 We mention that since the same pity in \cite{LO17} appears here, we need the non-resonance assumption for the embedded eigenvalues, too.

The main results of this paper are as follows.
\begin{theorem}\label{theorem1}
	Suppose that $V$ satisfies 
	\begin{align}
		V(x)=\frac{o(1)}{1+\abs{x}},
	\end{align}
	as $x\to\infty$ or $x\to-\infty$, then for any $\lambda\in\cup_{j=-\infty}^\infty(a_j,b_j)$, $\lambda$ is not an eigenvalue of $L$.
\end{theorem}

\begin{theorem}\label{theorem2}
	Suppose $\{\lambda_n\}_{n=1}^N\subset \cup_j(a_j,b_j)$ such that  quasimomenta $\{k(\lambda_n)\}_{n=1}^N$ are different. Suppose that for any $i,j\in\{1,2,\cdots,N\},$ $ k(\lambda_i)+k(\lambda_j)\neq \pi.$ Then there exist potentials $V\in C^\infty (\mathbb{R})$ with
	\begin{align}
		V(x)=\frac{O(1)}{1+\abs{x}}, \ \abs{x}\to\infty,\nonumber
	\end{align}
such that $\{\lambda_n\}_{n=1}^N$ are eigenvalues of $L$. 
\end{theorem}

\begin{corollary}
	Choose any band $(a_j,b_j)$. Let $e_j\in (a_j,b_j)$ with $k(e_j)=\frac{\pi}{2}$. 
	Suppose $\{\lambda_n\}_{n=1}^N$ is a finite set of distinct points in $(a_j,e_j)$ or $(e_j,b_j)$, such that  quasimomenta $\{k(\lambda_n)\}_{n=1}^N$ are different. Then there exist potentials $V\in C^\infty (\mathbb{R})$ with
	\begin{align}
		V(x)=\frac{O(1)}{1+\abs{x}}, \ \abs{x}\to\infty,\nonumber
	\end{align}
	such that $\{\lambda_n\}_{n=1}^N$ are eigenvalues of $L$. 
\end{corollary}

\begin{theorem}\label{theorem3}
	Suppose $\{\lambda_n\}_{n=1}^\infty\subset \cup_j(a_j,b_j)$ such that  quasimomenta $\{k(\lambda_n)\}_{n=1}^\infty$ are different. Suppose that for any $i,j, k(\lambda_i)+k(\lambda_j)\neq \pi.$ Then for any function $h(x)$ with $\lim_{\abs{x}\to\infty}\abs{h(x)}=\infty$, 
	there exist potentials $V\in C^\infty(\mathbb{R})$ with
	\begin{align}
	\abs{V(x)}\leq \frac{\abs{h(x)}}{1+\abs{x}},\nonumber
\end{align}
	such that $\{\lambda_n\}_{n=1}^\infty$ are eigenvalues of $L$. 
\end{theorem}

\section{Pr\"ufer transformation based on Floquet solution}

In this section, we  modify the Pr\"ufer transformation based on Floquet solution. 
In the following, we always assume that $\lambda\in \cup_{j=-\infty}^\infty (a_j,b_j)$.
Recall that $g=g(x,\lambda)$ is the Floquet solution of \eqref{periodicdiracequation} with $g(x+1,\lambda)=e^{ik(\lambda)}g(x,\lambda)$, one has 
\begin{align}\label{varphixperiodic}
	g(x,\lambda)=e^{ik(\lambda)x}h(x,\lambda),
\end{align}
where $h(x,\lambda)\in \mathbb{C}^{2}$ is a $1$-periodic function with respect to $x$.

Define the Wronskian 
\begin{align}
	W(y,g):&=y_1(x,\lambda)g_2(x,\lambda)-y_2(x,\lambda)g_1(x,\lambda)\nonumber,
\end{align}
one obtains
\begin{align}
	W(\bar{g},g)=2i{\rm{Im}}(\bar{g}_1(x,\lambda)g_2(x,\lambda)).\nonumber
\end{align}

\begin{lemma} For any  $\lambda\in \cup_j (a_j,b_j)$, one has
	\begin{align}\label{Wneq0}
		W(\bar{g},g)\neq 0.
	\end{align}
Therefore, for any $x\in \mathbb{R}$, we have $g_1(x,\lambda)\neq 0$ and $ g_2(x,\lambda)\neq 0$. 
\end{lemma}
	\begin{proof}
		Otherwise, we have 
		$g(x,\lambda)=c\bar{g}(x,\lambda)$ for some $c\neq 0.$ Then by \eqref{Floquetexponet} one has
		\begin{align}
			c\bar{g}(x+1,\lambda)=g(x+1,\lambda)=e^{ik(\lambda)}g(x,\lambda)=ce^{ik(\lambda)}\bar{g}(x,\lambda).\nonumber
		\end{align}
	Thus, by $\bar{g}(x+1,\lambda)=e^{-ik(\lambda)}\bar{g}(x,\lambda)$ one obtains the contradiction.
	\end{proof}
Let
\begin{align}
	W(\bar{g},g)=i\omega(\neq0).\nonumber
\end{align}
Denote by 
\begin{align}
g_j(x,\lambda)=\abs{g_j(x,\lambda)}e^{i\gamma_j(x,\lambda)}, \ j=1,2.\nonumber
\end{align}
By \eqref{varphixperiodic} one has that
\begin{align}\label{gammavarphi}
	\gamma_j(x,\lambda)=k(\lambda)x+\varphi_j(x,\lambda),\ j=1,2,
\end{align}
where $\varphi_j(x,\lambda)\mod 2\pi$ are $1$-periodic functions with respect to $x$.
\begin{proposition}\item[(a)] $2\abs{g_1(x,\lambda)}\abs{g_2(x,\lambda)}\sin(\gamma_2(x,\lambda)-\gamma_1(x,\lambda))=\omega.$
	\item[(b)]\begin{align}
		\gamma_1'(x,\lambda)=\frac{\omega(\lambda+p(x))}{2\abs{g_1(x,\lambda)}^2},\ \gamma_2'(x,\lambda)=\frac{\omega(\lambda-p(x))}{2\abs{g_2(x,\lambda)}^2}.\nonumber
	\end{align}
\end{proposition}
\begin{proof}
	\item[(a)] By Wronskian of $\bar{g}$ and $g$ we obtain
	\begin{align}
		2{\rm{Im}}(\bar{g}_1(x,\lambda)g_2(x,\lambda))=\omega,\nonumber
	\end{align}
hence one obtains the conclusion.
\item[(b)] By \eqref{periodicdiracequation} and 
\begin{align}
	\ln g_1(x,\lambda)=\ln\abs{g_1(x,\lambda)}+i\gamma_1(x,\lambda),\nonumber
\end{align}
one has
\begin{align}
	\gamma_1'(x,\lambda)&={\rm{Im}}\frac{g_1'(x,\lambda)}{g_1(x,\lambda)}\nonumber\\
	&={\rm{Im}}\frac{\lambda g_2(x,\lambda)-q(x)g_1(x,\lambda)+p(x)g_2(x,\lambda)}{g_1(x,\lambda)}\nonumber\\
	&=\frac{\lambda+p(x)}{\abs{g_1(x,\lambda)}^2}{\rm{Im}}(\bar{g}_1(x,\lambda)g_2(x,\lambda))\nonumber\\
	&=\frac{\omega(\lambda+p(x))}{2\abs{g_1(x,\lambda)}^2}.\nonumber
\end{align}
The calculation of $\gamma_2'(x,\lambda)$ is similar, we omit the steps.
\end{proof}

Now we introduce the modified Pr\"ufer transformation based on Floquet solution.

Let $y(x,\lambda)$ be the solution of \eqref{perturbeddiracequation}.
Define $\rho(x,\lambda)=\abs{\rho(x,\lambda)}^{i\eta(x,\lambda)}$ by 
\begin{align}
	\begin{pmatrix}\label{definitionofrho}
		y_1(x,\lambda)\\
		y_2(x,\lambda)
		\end{pmatrix}&=\frac{1}{2i}\left[\begin{pmatrix}
		\rho(x,\lambda)g_1(x,\lambda)-\bar{\rho}(x,\lambda)\bar{g_1}(x,\lambda)\\
		\rho(x,\lambda)g_2(x,\lambda)-\bar{\rho}(x,\lambda)\bar{g_2}(x,\lambda)
	\end{pmatrix}
	\right]\\
	&={\rm{Im}}
	\begin{pmatrix}
		\rho(x,\lambda)g_1(x,\lambda)\\
		\rho(x,\lambda)g_2(x,\lambda)
	\end{pmatrix}.\nonumber
\end{align}
We can normalize $\eta(x,\lambda)$ by $\eta(0,\lambda)\in (0,2\pi]$ and $\eta(x,\lambda)$ being continuous. Denote by $R(x,\lambda)=\abs{\rho(x,\lambda)}$.
Then one has
\begin{align}
	y_j(x,\lambda)=R(x,\lambda)\abs{g_j(x,\lambda)}\sin\theta_j(x,\lambda), j=1,2,
\end{align}
where $\theta_j(x,\lambda)=\eta(x,\lambda)+\gamma_j(x,\lambda).$ 
By \eqref{definitionofrho}, one can obtain
\begin{align}\label{rhoxW}
	\rho(x,\lambda)=\frac{2}{\omega}W(\bar{g},y).
\end{align}

\begin{proposition}
	For some $C>0$, we have
	\begin{align}\label{urcomparible}
		\frac{1}{C}\sqrt{\abs{y_1(x,\lambda)}^2+\abs{y_2(x,\lambda)}^2}\leq R(x,\lambda)\leq C\sqrt{\abs{y_1(x,\lambda)}^2+\abs{y_2(x,\lambda)}^2}.
	\end{align}
\end{proposition}
\begin{proof}
	Since
	\begin{align}
		\abs{W(\bar{g},y)}&=\abs{\bar{g_1}(x,\lambda)y_2(x,\lambda)-\bar{g_2}(x,\lambda)y_1(x,\lambda)}\nonumber\\
		&\leq \sqrt{\abs{g_1(x,\lambda)}^2+\abs{g_2(x,\lambda)}^2}\sqrt{\abs{y_1(x,\lambda)}^2+\abs{y_2(x,\lambda)}^2}.\nonumber
	\end{align}
Then by \eqref{rhoxW} and the periodicity of $\abs{g_1(x,\lambda)}^2+\abs{g_2(x,\lambda)}^2$, one has that for some $C>0$,
\begin{align}
	R(x,\lambda)\leq C\sqrt{\abs{y_1(x,\lambda)}^2+\abs{y_2(x,\lambda)}^2}.\nonumber
\end{align}
By \eqref{definitionofrho} and the periodicity of $\abs{g_1(x,\lambda)}^2+\abs{g_2(x,\lambda)}^2$, one directly obtains the left half.
\end{proof}

By \eqref{urcomparible}, one has that $y(\cdot,\lambda)\in L^2(\mathbb{R})\oplus L^2(\mathbb{R})$ is equivalent to that $R(\cdot,\lambda)\in L^2(\mathbb{R})$. Recall that the perturbation
\begin{align}
	V=\begin{pmatrix}
		V(x)&0\\
		0&-V(x)
	\end{pmatrix},\nonumber
\end{align}
we have the following Lemma.

\begin{lemma}\label{lemmaprufer}
	\item[(a)]\begin{align}
		\frac{\rho'(x,\lambda)}{\rho(x,\lambda)}=&\frac{2V(x)}{\omega}\abs{g_1(x,\lambda)}^2e^{-i\theta_1(x,\lambda)}\sin\theta_1(x,\lambda)\nonumber\\
		&-\frac{2V(x)}{\omega}\abs{g_2(x,\lambda)}^2e^{-i\theta_2(x,\lambda)}\sin\theta_2(x,\lambda).\nonumber
	\end{align}
\item[(b)]
\begin{align}
	\frac{R'(x,\lambda)}{R(x,\lambda)}=
	\frac{V(x)}{\omega}\bigg(\abs{g_1(x,\lambda)}^2\sin2\theta_1(x,\lambda)
	-\abs{g_2(x,\lambda)}^2\sin2\theta_2(x,\lambda)\nonumber\bigg).
\end{align}
\item[(c)] For $j=1,2$,
\begin{align}
	\theta_j'(x,\lambda)=&\gamma_j'(x,\lambda)\nonumber\\
	&-\frac{2V(x)}{\omega}\bigg(\abs{g_1(x,\lambda)}^2\sin^2\theta_1(x,\lambda)
	-\abs{g_2(x,\lambda)}^2\sin^2\theta_2(x,\lambda)\bigg).\nonumber
\end{align}
\end{lemma}
\begin{proof}
 By \eqref{perturbeddiracequation}, \eqref{periodicdiracequation} and \eqref{rhoxW}, one has
	\begin{align}
		\rho'(x,\lambda)
		=&\frac{2}{\omega}(\bar{g_1}'(x,\lambda)y_2(x,\lambda)+\bar{g_1}(x,\lambda)y_2'(x,\lambda)\nonumber\\
		&-\bar{g_2}'(x,\lambda)y_1(x,\lambda)-\bar{g_2}(x,\lambda)y_1'(x,\lambda))\nonumber\\
		=&\frac{2}{\omega}(V(x)\bar{g_1}(x,\lambda)y_1(x,\lambda)-V(x)\bar{g}_2(x,\lambda)y_2(x,\lambda)).\nonumber
	\end{align}
Therefore, by \eqref{definitionofrho} we can obtain (a). 

Taking the real part of 
\begin{align}\label{lnrholnr}
	(\ln \rho(x,\lambda))'=(\ln R(x,\lambda))'+i\eta'(x,\lambda),
\end{align}
one obtains (b). Taking the imaginary part of \eqref{lnrholnr}
and applying $\theta_j(x,\lambda)=\gamma_j(x,\lambda)+\eta(x,\lambda)$ one has (c).
\end{proof}

\begin{proof}[\textbf{Proof of Theorem \ref{theorem1}}]
	By the assumption in the theorem, we suppose that for some small $\varepsilon>0$, for any $x>x_0>0$, one has
	\begin{align}
		\abs{V(x)}\leq \frac{\varepsilon}{x}.\nonumber
	\end{align}
	Therefore, by (b) of Lemma \ref{lemmaprufer} , one has
	\begin{align}
		\ln R(x,\lambda)-\ln R(x_0,\lambda)&\geq -C\varepsilon \int_{x_0}^x\frac{1}{t}dt\nonumber\\
		&= -C\varepsilon \ln \frac{x}{x_0},\nonumber
	\end{align}
	where $C>0$ so that for any $x\in\mathbb{R}$, $\abs{g_1(x,\lambda)}^2+\abs{g_2(x,\lambda)}^2<\omega C$. Let $\varepsilon$ be small so that $C\varepsilon<\frac{1}{2}$. One has that for any large enough $x$, 
	\begin{align}
		R(x,\lambda)\geq x^{-\frac{1}{2}}.\nonumber
	\end{align}
	Therefore, $R(\cdot,\lambda)\notin L^2(\mathbb{R})$. Hence, for any $\lambda\in \cup_{j=1}^\infty (a_j,b_j)$, $\lambda$ is not an eigenvalue of $L$.
\end{proof}
\section{Preparations}

Denote by
\begin{align}
	\Gamma_1(x,\lambda):=2\gamma_2(x,\lambda)-2\gamma_1(x,\lambda),\nonumber
\end{align}
by \eqref{gammavarphi}, we know that $\Gamma_1(x,\lambda)\mod2\pi$ is a $1$-periodic function in $x$. By the defintion of $\theta(x,\lambda)$, one has
\begin{align}
	2\theta_2(x,\lambda)-2\theta_1(x,\lambda)=\Gamma_1(x,\lambda).\nonumber
\end{align}
Therefore, 
\begin{align}
	&\abs{g_1(x,\lambda)}^2\sin2\theta_1(x,\lambda)
	-\abs{g_2(x,\lambda)}^2\sin2\theta_2(x,\lambda)\nonumber\\
	=&\abs{g_1(x,\lambda)}^2\sin2\theta_1(x,\lambda)-\abs{g_2(x,\lambda)}^2\sin(2\theta_1(x,\lambda)+\Gamma_1(x,\lambda))\nonumber\\
	\label{middleequation}=&\sin2\theta_1(x,\lambda)\left(\abs{g_1(x,\lambda)}^2-\abs{g_2(x,\lambda)}^2\cos\Gamma_1(x,\lambda)\right)\\
	&-\abs{g_2(x,\lambda)}^2\sin\Gamma_1(x,\lambda)\cos2\theta_1(x,\lambda).\nonumber
\end{align}
Denote by 
\begin{align}
	\Psi(x,\lambda)=&\bigg(\left(\abs{g_1(x,\lambda)}^2-\abs{g_2(x,\lambda)}^2\cos\Gamma_1(x,\lambda)\right)^2\nonumber\\
	&+\left(\abs{g_2(x,\lambda)}^2\sin\Gamma_1(x,\lambda)\right)^2\bigg)^{\frac{1}{2}}\nonumber\\
	=&\left(\abs{g_1(x,\lambda)}^4+\abs{g_2(x,\lambda)}^4-2\abs{g_1(x,\lambda)}^2\abs{g_2(x,\lambda)}^2\cos\Gamma_1(x,\lambda)\right)^{\frac{1}{2}}\nonumber.
\end{align}
Clearly, $\Psi(x,\lambda)$ is  a $1$-periodic function in $x$. We mention that for any $x\in\mathbb{R}$,  $\Psi(x,\lambda)> 0.$
Otherwise, one has
\begin{align}
	\abs{g_1(x,\lambda)}^4+\abs{g_2(x,\lambda)}^4-2\abs{g_1(x,\lambda)}^2\abs{g_2(x,\lambda)}^2\cos\Gamma_1(x,\lambda)=0.\nonumber
\end{align}
Then we can obtain 
\begin{align}
	\abs{g_1(x,\lambda)}=\abs{g_2(x,\lambda)},\nonumber
\end{align}
and 
\begin{align}
	\gamma_2(x,\lambda)-\gamma_1(x,\lambda)=\frac{\Gamma_1(x,\lambda)}{2}=n_0\pi,\nonumber
\end{align}
for some $n_0\in\mathbb{Z}$. Then we have
\begin{align}
	g_1(x,\lambda)=\pm g_2(x,\lambda),\nonumber
\end{align}
which can not happen by the fact $W(\bar{g},g)\neq 0$. Hence, by \eqref{middleequation}, we have
\begin{align}
	&\abs{g_1(x,\lambda)}^2\sin2\theta_1(x,\lambda)
	-\abs{g_2(x,\lambda)}^2\sin2\theta_2(x,\lambda)\nonumber\\
	=&\Psi(x,\lambda)\sin (2\theta_1(x,\lambda)+\Gamma_2(x,\lambda)),\label{middleequationnew}
\end{align}
where $\Gamma_2(x,\lambda)\mod 2\pi$ is also a $1$-periodic function in $x$, and 
\begin{align}
	\sin\Gamma_2(x,\lambda)&=-\frac{\abs{g_2(x,\lambda)}^2\sin\Gamma_1(x,\lambda)}{\Psi(x,\lambda)},\nonumber\\
	\cos\Gamma_2(x,\lambda)&=\frac{\abs{g_1(x,\lambda)}^2-\abs{g_2(x,\lambda)}^2\cos\Gamma_1(x,\lambda)}{\Psi(x,\lambda)}.\nonumber
\end{align}
Let $\xi(x,\lambda)=2\theta_1(x,\lambda)+\Gamma_2(x,\lambda)$,  
one obtains
\begin{lemma} 
We have
\begin{align}\label{pruferR}
		\frac{R'(x,\lambda)}{R(x,\lambda)}=\frac{V(x)}{\omega}\Psi(x,\lambda)\sin \xi(x,\lambda),
	\end{align}
and
\begin{align}
	\xi'(x,\lambda)=&2k(\lambda)+\delta'(x,\lambda)\nonumber\\
	&-\frac{2V(x)}{\omega}\left(\abs{g_1(x,\lambda)^2}-\abs{g_2(x,\lambda)}^2-\Psi(x,\lambda)\cos\xi(x,\lambda)\right),\label{pruferxi}
\end{align}
where $\delta(x,\lambda)=2\varphi_1(x,\lambda)+\Gamma_2(x,\lambda)$ and $\delta(x,\lambda)\mod2\pi$ is a $1$-periodic function in $x$.
\end{lemma}
\begin{proof}
	One directly obtains \eqref{pruferR} by (b) of Lemma \ref{lemmaprufer} and \eqref{middleequationnew}. By \eqref{gammavarphi}, \eqref{middleequationnew} and (c) of Lemma \ref{lemmaprufer}, one can obtain \eqref{pruferxi}.
\end{proof}

\section{Oscillatory Intergral estimate}

\begin{lemma}\cite[Lemma 3.1]{L19stark}\label{L}
	Let $\beta_1>0,\beta_2>0$ and $a\neq 0$ be constants. Suppose $\beta_1+\beta_2>1,\beta_2>\frac{1}{2}$.
	Suppose that $\theta(x)$ satisfies 
	\begin{align}
		\theta'(x)=a+\frac{O(1)}{1+\abs{x}^{\beta_1}}, \ \abs{x}\to\infty.
	\end{align}
Let $\beta=\min\{\beta_2,\beta_1+\beta_2-1,2\beta_2-1\}$. Then for any $x\geq x_0>1$, one has
\begin{align}
	\int_{x_0}^x\frac{\sin\theta(t)}{t^{\beta_2}}dt=\frac{O(1)}{x_0^{\beta}},
\end{align}
and
\begin{align}
		\int_{x_0}^x\frac{\cos\theta(t)}{t^{\beta_2}}dt=\frac{O(1)}{x_0^\beta}.
\end{align}
\end{lemma}

The following Lemma has been essentially proved in \cite[Prop. 5.1]{LO17}, for completeness, we provide with the proof by using Lemma \ref{L}.
\begin{lemma}\label{L2}
	Assume that $\gamma(x)\mod 2\pi,\Gamma(x)$ are continuous $1$-periodic functions on $\mathbb{R}$. Let $a\in \mathbb{R}\setminus 2\pi\mathbb{Z}$. Suppose that $\theta(x)$ satisfies
	\begin{align}
		\theta'(x)=a+\gamma'(x)+\frac{O(1)}{\abs{x}},\ \abs{x}\to \infty,
	\end{align} 
then one has for any $x>x_0>1$, 
\begin{align}\label{Gammasintheta}
	\int_{x_0}^x\frac{\Gamma(t)\sin\theta(t)}{t}dt=\frac{O(1)}{x_0},
\end{align}
and for any $x<x_0<-1$,
\begin{align}\label{Gammasinthetamines}
	\int_{x}^{x_0}\frac{\Gamma(t)\sin\theta(t)}{{t}}dt=\frac{O(1)}{{x_0}}.
\end{align}
\end{lemma}
\begin{proof}We only give the proof of \eqref{Gammasintheta}, then one can obtain \eqref{Gammasinthetamines} by letting $t=-s$. Denote by 
\begin{align}
	\tilde{\theta}(x)=\theta(x)-\gamma(x).\nonumber
\end{align}
Then one has
\begin{align}\label{thetatilde}
	\tilde{\theta}'(x)=a+\frac{O(1)}{\abs{x}}, \ \abs{x}\to\infty,
\end{align}
and
\begin{align}
	\sin\theta(x)=\cos\gamma(x)\sin\tilde{\theta}(x)+\sin\gamma(x)\cos\tilde{\theta}(x).\nonumber
\end{align}
Hence, we only need to prove
\begin{align}\label{Gammasinthetaminus}
	\int_{x_0}^{x}\frac{\Gamma(t)\cos\gamma(t)\sin\tilde{\theta}(t)}{{t}}dt=\frac{O(1)}{x_0},
\end{align}
and
\begin{align}\label{Gammacosthetaminus}
	\int_{x_0}^{x}\frac{\Gamma(t)\sin\gamma(t)\cos\tilde{\theta}(t)}{{t}}dt=\frac{O(1)}{x_0}.
\end{align}
To avoid repetition, we only prove \eqref{Gammasinthetaminus}. We mention that $\Gamma(x)\cos\gamma(x)$ is still $1$-periodic. 
Consider the Fourier expansion of $\Gamma(x)\cos\gamma(x)$,
	\begin{align}
		\Gamma(x)\cos\gamma(x)=\frac{a_0}{2}+\sum_{n=1}^\infty a_n\cos(2\pi nx)+b_n\sin(2\pi nx).\nonumber
	\end{align}
To prove \eqref{Gammasinthetaminus}, by the fact that $\sum a_n^2+b_n^2<\infty,$ we only need to show 
\begin{align}
	\int_{x_0}^x\frac{\sin\tilde{\theta}(t)}{t}dt=\frac{O(1)}{x_0},\nonumber
\end{align}
and for any $n>0$,
\begin{align}
	\int_{x_0}^x\frac{\cos(2\pi nt)\sin\tilde{\theta}(t)}{t}dt=\frac{1}{n}\frac{O(1)}{x_0},\nonumber
\end{align} 
and 
\begin{align}\label{sinsin}
		\int_{x_0}^x\frac{\sin(2\pi nt)\sin\tilde{\theta}(t)}{t}dt=\frac{1}{n}\frac{O(1)}{x_0}.
\end{align}
By the same reason, we only prove \eqref{sinsin}. Since
\begin{align}
\sin(2\pi nt)\sin\tilde{\theta}(t)=\frac{1}{2}\left(\cos(2\pi nt-\tilde{\theta}(t))-\cos(2\pi nt+\tilde{\theta}(t))\right),\nonumber
\end{align}
we have
\begin{align}
	\int_{x_0}^x\frac{\sin(2\pi nt)\sin\tilde{\theta}(t)}{t}dt=&\int_{x_0}^x\frac{\cos(2\pi nt-\tilde{\theta}(t))}{2t}dt\nonumber\\
	&-\int_{x_0}^x\frac{\cos(2\pi nt+\tilde{\theta}(t))}{2t}dt.\nonumber
\end{align}
Change variables with $s=nt$, by \eqref{thetatilde} and the fact that $a\in \mathbb{R}\setminus 2\pi\mathbb{Z}$, applying Lemma \ref{L} with $\beta_1=\beta_2=1$, one obtains the result.
\end{proof}
\section{Construction}
We assume that $\lambda$ and $\lambda_j$ are different values in $\cup_{j=-\infty}^\infty(a_j,b_j)$ with $k(\lambda),k(\lambda_j)\neq \frac{\pi}{2}$, $k(\lambda)\neq k(\lambda_j)$ and $k(\lambda)+k(\lambda_j)\neq \pi$. Denote by
\begin{align}
	k&=k(\lambda),\ k_j=k(\lambda_j), \ \delta(x)=\delta(x,\lambda), \delta_j(x)=\delta(x,\lambda_j),\nonumber\\
	R(x)&=R(x,\lambda)\ , R_j(x)=R(x,\lambda_j),\ \xi(x)=\xi(x,\lambda),\ \xi_j(x)=\xi(x,\lambda_j)\nonumber\\
	\Psi(x)&=\Psi(x,\lambda),\ \Psi_j(x)=\Psi(x,\lambda_j).\nonumber
\end{align}

We mention that $\delta(x),\Psi(x),\delta_j(x),\Psi_j(x)$ are independent of $V$. 
Next, we will give the construction of potentials $V$ on $[a,\infty)$ and on $(-\infty,a]$ for some large $a>0$,  the potential on $(-a,a)$ will not influence the speed of the decreases of $R$ and $R_j$. Consider the nonlinear differential equation for $x>b>0$,
\begin{align}\label{differentialeuqation326}
	\xi'(x,\lambda,a,b,\xi_0)=&2k+\delta'(x)\\
	&+\frac{2C\sin\xi(x)}{x-b}\left(\abs{g_1(x,\lambda)}^2-\abs{g_2(x,\lambda)}^2-\Psi(x)\cos\xi(x)\right),\nonumber
\end{align}
and for $x<-b$,
\begin{align}\label{differentialeuqation3262}
	\xi'(x,\lambda,a,b,\xi_1)=&2k+\delta'(x)\\
	&+\frac{2C\sin\xi(x)}{x+b}\left(\abs{g_1(x,\lambda)}^2-\abs{g_2(x,\lambda)}^2-\Psi(x)\cos\xi(x)\right),\nonumber
\end{align}
where $C>0$ is a large constant will be defined later. Solve \eqref{differentialeuqation326} on $[a,\infty)$ with the initial condition $\xi(a)=\xi_0$, where $a>b$, we can obtain a unique solution $\xi(x)=\xi (x,\lambda,a,b,\xi_0)$ on $[a,\infty)$. Solve \eqref{differentialeuqation3262} on $(-\infty,a]$ with the initial condition $\xi(-a)=\xi_1$, we can obtain a unique solution $\xi(x)=\xi (x,\lambda,a,b,\xi_1)$ on $(-\infty,a]$.
Let
\begin{align}\label{potentialconstruction326}
	V(x,\lambda,a,b,\xi_0)=-\frac{\omega C}{x-b}\sin\xi (x,\lambda,a,b,\xi_0), \ x\in [a,\infty),
\end{align}
and
\begin{align}\label{potentialconstruction3262}
	V(x,\lambda,a,b,\xi_1)=-\frac{\omega C}{x+b}\sin\xi (x,\lambda,a,b,\xi_1), \ x\in (-\infty,a].
\end{align}
Denote by 
\begin{align}
	V(x,\lambda,a,b,\xi_0,\xi_1)=\begin{cases}
		V(x,\lambda,a,b,\xi_0), \ x\in [a,\infty),\\	V(x,\lambda,a,b,\xi_1), \ x\in (-\infty,-a].
	\end{cases}\label{potential}
\end{align}
\begin{lemma}\label{lemma-100}
	Suppose $k\neq \frac{\pi}{2}$. Let $V$ be difined by \eqref{potential}, then for any $\pm x> a$, we have
\begin{align}\label{RpmxC}
	\ln R(x)-\ln R(\pm a)\leq -100\ln\frac{\pm x-b}{a-b}+C,
\end{align}
and 
\begin{align}
	\ln R(x)\leq \ln R(\pm a).
\end{align}
\end{lemma}
\begin{proof}Without loss of generality, assume $b=0$, and we only consider the case $x>a$, since we can directly obtain the result by letting $t=-x$.
	By \eqref{pruferR} and \eqref{potentialconstruction326}, one has
	\begin{align}
		\ln R(x)-\ln R(a)=&\int_a^x-\frac{C\Psi(t)}{t}\sin^2\xi(t)dt\nonumber\\
		=&\int_a^x-\frac{C\Psi(t)}{2t}\left(1-\cos2\xi(t)\right)\nonumber\\
		=&\int_a^x-\frac{C\Psi(t)}{2t}dt+\int_a^x\frac{C}{2}\frac{\Psi(t)}{t}\sin\left(2\xi(t)+\frac{\pi}{2}\right)dt.\label{lnr-lnra}
	\end{align}
By \eqref{differentialeuqation326}, we have
\begin{align}
	\left(2\xi(t)+\frac{\pi}{2}\right)'=4k+2\delta'(t)+\frac{O(1)}{t},\nonumber
\end{align}
as $t\to\infty.$
Then applying \eqref{Gammasintheta} one obtains
\begin{align}
	\int_a^x\frac{C}{2}\frac{\Psi(t)}{t}\sin\left(2\xi(t)+\frac{\pi}{2}\right)dt=O(1),\nonumber
\end{align}
 by \eqref{lnr-lnra} one has 
\begin{align}
	\ln R(x)-\ln R(a)=O(1)-\int_a^x\frac{C\Psi(t)}{2t}dt.\nonumber
\end{align}
Therefore, the result follows from that $\Psi(x)>0$ and the periodicity of $\Psi(x)$.
\end{proof}

\begin{lemma}\label{lemma1.5}
	Suppose that $k\neq k_j$, $k+k_j\neq \pi$. Let $V$ be defined by \eqref{potential}, then for any $\pm x>x_0\geq a$ and large enough $x_0-b$ one has
	\begin{align}\label{Rj1.5R0}
		R_j(x)\leq 1.5 R_j(\pm x_0).
	\end{align}
\end{lemma}
\begin{proof}
	We only show the case of $x>x_0\geq a$, then we can obtain the result by letting $t=-x$. By \eqref{pruferxi} and \eqref{potentialconstruction326}, one has
	\begin{align}
		&\ln R_j(x)-\ln R_j(x_0)\nonumber\\
		=&\int_{x_0}^x-\frac{C\Psi_j(t)}{t-b}\sin\xi(t)\sin\xi_j(t)dt\nonumber\\
		=&\int_{x_0}^x-\frac{C\Psi_j(t)}{2(t-b)}\cos(\xi(t)-\xi_j(t))dt+\int_{x_0}^x\frac{C\Psi_j(t)}{2(t-b)}\cos(\xi(t)+\xi_j(t))dt.\nonumber
	\end{align}
By \eqref{pruferxi}, \eqref{differentialeuqation326} and \eqref{potentialconstruction326}, we have
\begin{align}
	\left(\xi(t)\pm \xi_j(t)+\frac{\pi}{2}\right)'=2k\pm2k_j+\delta'(t)\pm\delta'_j(t)+\frac{O(1)}{t-b},\nonumber
\end{align}
then by \eqref{Gammasintheta} one obtains the result.
\end{proof}
\begin{proposition}\label{prop5.3}
	Let $\lambda$ and $A=\{\lambda_j\}_{j=1}^N$ be in $\cup_{j=-\infty}^\infty(a_j,b_j)$. Suppose that $k\neq k_j$, $k\neq \frac{\pi}{2}$ and $k+k_j\neq \pi$ for any $j\in\{1,2,\cdots,N\}$. Suppose $\xi_0,\xi_1\in [0,\pi]$. Let $x_1>x_0>b$. Then there exist constants $K(\lambda,A), C(\lambda,A)$ (independent of $x,x_0$ and $b$) and the potential $\tilde{V}(x,\lambda,x_0,x_1,b,\xi_0,\xi_1)$ such that for $x_0-b>K(\lambda,A)$ the following holds:
	\item[\textbf{Potential}:]$\tilde{V}\in C_0^{\infty}((x_0,x_1)\cup (-x_1,-x_0))$, and
	\begin{align}\label{potentialC}
		\abs{\tilde{V}(x,\lambda,A,x_0,x_1,b,\xi_0,\xi_1)}&\leq \frac{C(\lambda,A)}{\pm x-b},\ x_0<\pm x<x_1.
	\end{align}
\item[\textbf{Solution for} $\lambda$:] 
\begin{align}\label{R-100C}
	R(\pm x_1)\leq C(\lambda,A)\left(\frac{x_1-b}{x_0-b}\right)^{-100}R(\pm x_0),
\end{align}
and for $x_0<\pm x<x_1$,
\begin{align}\label{Rx2x0}
	R(x)\leq 2R(\pm x_0).
\end{align}
\item[\textbf{Solution for} $\lambda_j$:] for any $x_0<\pm x\leq x_1$,
\begin{align}\label{Rjx2-x0}
	R_j(x)\leq 2R_j(\pm x_0).
\end{align}
\end{proposition}
\begin{proof}
	Let $V$ be defined by  \eqref{potential} with $a=x_0$. Applying Lemmas \ref{lemma-100} and \ref{lemma1.5} with $x=x_1$ and $a=x_0$, one can obtain \eqref{R-100C}, \eqref{Rx2x0} and \eqref{Rjx2-x0} with the coefficients $2$ being replaced by $1.5$. We need to modify the potential $V$ so that $V\in C_0^\infty$. By \eqref{pruferR} and \eqref{pruferxi} we know that a small change on $V$ will not cause much influence on $R$ and $R_j$ on $[x_0,x_1]\cup [-x_1,-x_0]$, thus Lemmas \ref{lemma-100} and \ref{lemma1.5} still hold, and we can obtain the result.
\end{proof}

\begin{proof}[\textbf{Proof of Theorems \ref{theorem2} and \ref{theorem3}}.]
	Once we obtain Proposition \ref{prop5.3}, we can construct potentials for Theorems \ref{theorem2} and \ref{theorem3} step by step, which has been well established in \cite{JL19,L19stark,LO17}. 
	
	We only give an outline of the construction here. For example, we want $\lambda$ and $\lambda_1$ to be embedded eigenvalues, then we need to construct potentials so that $R(x)$ and $R_1(x)$ decrease quickly enough to be $L^2$ near $\pm\infty$. First of all, let $V$ be in Propsition \ref{prop5.3} on some intervals $(-T_1,T_0]\cup [T_0,T_1)$, then by \eqref{R-100C} and \eqref{Rjx2-x0} we know that $R(x)$ decreases very fast and $R_1(x)$ will not undergo a significant increment. Next let $V$ be in Propsition \ref{prop5.3} on some intervals $(-T_2,T_1]\cup [T_1,T_2)$ with $\lambda$ and $\lambda_1$ exchanged, then by \eqref{R-100C} and \eqref{Rjx2-x0} we know that $R_1(x)$ will decreases very fast and $R(x)$ will not undergo a significant increment. By choosing intervals $(-T_{j+1},T_j]\cup [T_j,T_{j+1})$ properly, one can obtain that $R(x)$ and $R_1(x)$ are $L^2$ near $\pm\infty$.
	
	For many embedded eigenvalues. Let $\{N_r\}_{r\in\Z^+}$ be a non-decreasing sequence which goes to infinity slowly depending on $h(x)$. We further assume $N_{r+1}=N_r+1$ when $N_{r+1}>N_r$. At the $r$th step, we take $N_r$ eigenvalues $\{\lambda_1,\lambda_2,\cdots,\lambda_{N_r-1},\lambda_{N_r}\}$ into consideration. Applying Proposition \ref{prop5.3}, we construct potentials with $N_r$ pieces, where each piece comes from \eqref{potential} with $\lambda$ being an eigenvalue.
	The main difficulty is to control the size $T_r-T_{r-1}$ of each piece. The construction in \cite{JL19,LO17,L19stark} only uses inequalities \eqref{potentialC}, \eqref{R-100C}, \eqref{Rx2x0} and \eqref{Rjx2-x0} to obtain appropriate $T_r-T_{r-1}$ and $N_r$. Hence Proposition \ref{prop5.3} implies Theorems \ref{theorem2} and \ref{theorem3}. 
\end{proof}

	\section*{Acknowledgments}  The authors are supported by the National Natural
	Science Foundation of China (11871031).

	\bibliographystyle{abbrv} 
	\bibliography{reference}

\end{document}